\documentclass[12pt,a4paper,reqno]{article}
\usepackage{graphics}
\usepackage{epsfig}

\usepackage{amsmath,amsthm,amssymb}
\newtheorem{theorem}{Theorem}[section]

\newtheorem{proposition}[theorem]{Proposition}

\theoremstyle{definition}

\theoremstyle{remark}

\numberwithin{equation}{section}

\begin{document}
\title{On Codes over $\mathbb{Z}_{p^2}$ and its Covering Radius
\thanks{The first author would like to thank the Department of Science and Technology (DST), New Delhi, India for their financial support
in the form of INSPIRE Fellowship
(DST Award Letter No.IF130493/DST/INSPIRE Fellowship / 2013/(362) Dated: 26.07.2013) 
to carry out this work.}
}
\author{ N. Annamalai\\
Research Scholar\\
Department of Mathematics\\
Bharathidasan University\\
Tiruchirappalli-620 024, Tamil Nadu, India\\
{Email: algebra.annamalai@gmail.com}
\bigskip\\
  C.~Durairajan \\
Assistant Professor  \\
Department of Mathematics\\ 
Bharathidasan University\\
Tiruchirappalli-620024, Tamil Nadu, India\\
{Email: cdurai66@rediffmail.com   }
\hfill \\
\hfill \\
\hfill \\
\hfill \\
{\bf Proposed running head:}On Codes over $\mathbb{Z}_{p^2}$ and its Covering Radius}
\date{}
\maketitle

\newpage

\vspace*{0.5cm}
\begin{abstract}This paper gives lower and upper bounds on the covering radius of codes over
$\mathbb{Z}_{p^2}$ with respect to Lee distance. We also determine the covering radius of various
Repetition codes over $\mathbb{Z}_{p^2}.$

\end{abstract}
\vspace*{0.5cm}


\vspace{0.5cm}
\vspace{1.5cm}

\noindent
Corresponding author:\\ 
\\
 \hspace{1cm} 
Dr. C. Durairajan\\
\noindent
Assistant Professor\\
Department of Mathematics \\
Bharathidasan University\\
Tiruchirappalli-620024, Tamil Nadu, India\\
\noindent
E-mail: cdurai66@rediffmail.com
\newpage

\section{Introduction}
In last decade, there has been a burst of research activities  in codes over finite rings. In particular, codes over $\mathbb{Z}_{p^s}$ and 
$\mathbb{Z}_4$ received much attention \cite{psole}, \cite{bonne}, \cite{gull},\cite{DHP}, \cite{GH}, \cite{kumar},
\cite{hamm} and \cite{har}. One of the important parameter of a code is the covering radius
of the code.

The covering radius of binary linear codes is a widely studied parameter \cite{kar}. 
Recently, the covering radius of codes over $\mathbb{Z}_4$ has been investigated with respect to Lee and Euclidean distances 
and several upper and lower bounds on the covering radius have been obtained \cite{aok} and  \cite{manish}.

A {\it linear code} $C$ of length $n$ over a finite ring $R$ is an additive subgroup of $R^{n}.$ An element of $C$ is
called a {\it codeword} of $C$ and a {\it generator matrix} of $C$ is a matrix whose rows generate $C.$ 
Two codes are said to be {\it equivalent} if one can be obtained from the other by permuting the coordinates and (if necessary) changing the
signs of certain coordinates. Code differing by only a permutation of coordinates of the code is 
called a {\it permutation equivalent}. 

The {\it Hamming weight} $w_{H}(x)$ of a
vector $x$ in $R^n$ is the number of non-zero coordinates in $x.$ 
The Lee weight of $x\in\mathbb{Z}_{p^2}$ is defined by
\begin{equation*}
w_L(x)=
 \begin{cases}
x\,\, & \text{if}\,\, 0\leq x \leq p\\
p \,\, &\text{if}\,\, p+1 \leq x \leq p^2-p\\
p^2-x \,\,&\text{if} \,\, p^2-p+1\leq x <p^2.
 \end{cases}
\end{equation*} 
Let $x, y\in R^n,$ then the Hamming distance $d_H(x, y)$ and the Lee distances $d_L(x, y)$ between two vectors $x$ and $y$
 are $w_H(x - y)$ and $w_L(x - y),$ respectively. The minimum Hamming(Lee) distance $d_H$ of a code $C$ is the smallest Hamming(Lee) distance
 between any two distinct  codewords of $C.$

In this paper, we study the covering radius of codes over  $\mathbb{Z}_{p^2}.$ In section 2 contains some preliminaries and 
basic results for the covering radius of codes over $\mathbb{Z}_{p^s}.$ In Section 3, we determine the covering
radius of different types of repetition codes over $\mathbb{Z}_{p^2}.$ 
\section{Preliminaries}
Let $d$ be the general distance out of various possible distance such as Hamming distance and Lee distance. The
covering radius of a code $C$ over a finite ring $R$  with respect to a general distance $d$ is given by

\begin{equation*} 
r_d (C) = \max\limits_{u\in R^n}\bigg\{ \min\limits_{c\in C} d(u,c) \bigg \}.
\end{equation*}
It is easy to see that $r_d (C)$ is the least non-negative integer $r_d$ such that 

\begin{equation*}
\mathbb{Z}_{p^s}^n= \bigcup\limits_{c\in C}S_{r_d}(c) \text{ where } 
S_{r_d}(u) = \{v \in\mathbb{Z}_{p^s}^n\mid d(u, v) \leq r_d\}.
\end{equation*} 
If $C=R^n,$ then the covering radius of $C$ is $0.$ If $C=\{0\},$ then the covering radius of $C$ is $n.$

Let $C$ be a linear code, then the translate $u + C = \{u + c \mid c \in C\}$ is called the {\it coset} of $C$ where $u\in\mathbb{Z}_{p^s}.$
A minimum weight word in a coset is called a {\it coset leader.} The following
proposition is straightforward  from the definition of covering radius and the coset leader.

\begin{proposition}\label{c}
 The covering radius of a linear code $C$ with respect to the general distance $d$ is the largest weight among all coset leaders.
\end{proposition}

In \cite{yil}, the Gray map 
$\phi : \mathbb{Z}_{p^2}^n\rightarrow \mathbb{Z}_p^{pn}$ is the coordinates-wise extension of the function from
$\mathbb{Z}_{p^2}$ to $\mathbb{Z}_{p}^p$ defined by

\begin{equation*}
\begin{split}
 0 &\rightarrow 000\cdots000\\
 1& \rightarrow 100\cdots000\\
 2&\rightarrow 110\cdots000\\
&\vdots\\
p&\rightarrow 111\cdots111\\
 p + 1&\rightarrow 211\cdots111\\ 
&\vdots\\
2p&\rightarrow 222\cdots222\\
&\vdots\\
(p - 1)p&\rightarrow (p - 1)(p-1)(p-1)\cdots(p-1)(p-1)(p - 1)\\
 (p - 1)p + 1&\rightarrow  0(p - 1)(p-1)\cdots(p-1)(p-1)(p - 1)\\
&\vdots\\
 p^2 - 1&\rightarrow 000\cdots00(p - 1). 
\end{split}
\end{equation*}

 The image $\phi(C)$ of a linear code $C$ over $\mathbb{Z}_{p^2}$  by the Gray map is a $p$-ary code of length $pn.$ Note that
 the minimum weight and the minimum distance of $C$ with respect to Lee is the same as the minimum weight and minimum distance of 
 $\phi(C)$ with respect to Hamming. Therefore, we have
\begin{proposition}\label{d}
Let $C$ be a code over $\mathbb{Z}_{p^2}$  and $\phi(C)$ the  Gray image of $C.$ Then $r_L(C) = r_H(\phi(C))$ where $\phi(C)$ is the Gray image of $C.$
\end{proposition}
 A code of length $n$ with $M$ codewords and the minimum Hamming distance 
$d_H,$ the minimum Lee distance $d_L$ over $\mathbb{Z}_{p^2}$ is called an  $(n, M, d_H, d_L)$ code over $\mathbb{Z}_{p^2}$  is said to be  type $\alpha$ (type $\beta$) code if 
 $d_H=\bigg\lceil\dfrac{d_L}{p}\bigg\rceil \bigg(d_H>\bigg\lceil\dfrac{d_L}{p}\bigg\rceil\bigg).$
 
The dual code $C^{\perp}$  of $C$ is defined as $\{x\in\mathbb{Z}_{p^2}^n \mid x\cdot y = 0\,\, \text{for all}\,\, y\in C\}$ where 
$x\cdot y=\sum\limits_{i=1}^{n}x_i y_i\, (mod\,\, p^2).$

Now, we give a few known lower and upper bounds on the covering radius of codes over $\mathbb{Z}_{p^2}$
with respect to Lee weight. The proof of the following Proposition \ref{e} and Theorem \ref{f} is similar to
that of Proposition 3.2 and Theorem 4.5 in \cite{aok} and is omitted.

\begin{proposition}\label{e}
For any code $C$ of length $n$ over $\mathbb{Z}_{p^2},$ 
\begin{equation*}
\dfrac{p^{pn}}{|C|}\leq \sum\limits_{i=0}^{r_L(C)}\binom{pn}{i}.
\end{equation*}
\end{proposition}
Let $C$ be a code over $\mathbb{Z}_{p^2}$ and let
\begin{equation*}
s(C^{\perp}) = \#\{i | A_i(C^{\perp}) \neq 0, i \neq 0\}
\end{equation*}
where $A_i(C^{\perp})$ is the number of codewords of Lee weight $i$ in $C^{\perp}.$

\begin{theorem}\label{f} 
Let $C$ be a code over $\mathbb{Z}_{p^2},$ then $r_L(C)\leq s(C^{\perp}).$
\end{theorem}

The following result is generalization of Mattson \cite{kar} results of code over a finite field to a finite ring.

\begin{theorem}\label{g}
 If $C_0$ and $C_1$ are codes over $\mathbb{Z}_{p^2}$ generated by matrices $G_0$ and
$G_1$ respectively and if $C$ is the code generated by
\begin{equation*}
G =
\begin{pmatrix}
0&\vline&G_1\\
\hline
G_0&\vline &A
\end{pmatrix},
\end{equation*}
then $r_d(C)\leq r_d(C_0) + r_d(C_1)$ and the covering radius of the concatenation of $C_0$ and
$C_1$ is greater than or equal to $r_d(C_0) + r_d(C_1)$ for all distances $d$ over $\mathbb{Z}_{p^2}.$
\end{theorem}
The proof of this is similar to that of Mattson \cite{kar} and hence omitted.

\begin{theorem}\cite{kar}\label{cb}
Let C be the Cartesian product of two codes $C_1$ and $C_2,$ then the covering radius of C is  $r(C_1)=r(C_2)+r(C_3).$  
\end{theorem}
\section{Repetition Codes}
 A $q$-ary repetition code $C$ over a finite field 
 $\mathbb{F}_{q}=\{\alpha_0=0,\alpha_1=1,\alpha_2,\cdots,\alpha_{q-1}\}$ is an $[n,1,n]$ code 
 $C=\{\bar{\alpha}\mid \alpha\in \mathbb{F}_{q}\}$ where $\bar{\alpha}=\alpha\alpha\cdots\alpha\in\mathbb{F}_{q}^n.$ The covering radius of $C$ is 
 $\left\lceil\dfrac{n(q-1)}{q}\right\rceil$\,\cite{cd}. The block repetition  code generated by 
 
$$G=\bigg[\overbrace{11\cdots1}^n \overbrace{\alpha_2\alpha_2\cdots\alpha_2}^n\cdots\overbrace{\alpha_{q-1}\alpha_{q-1}\cdots\alpha_{q-1}}^n\bigg]$$
is equivalent to the repetition code of length $(q-1)n$ over $\mathbb{F}_q.$ Then by the above, the covering radius of the code generated by $G$ is 
$\left\lceil \dfrac{n(q-1)^2}{q}\right\rceil.$ 
 
 If we consider the repetition code over a finite ring $R,$ then there are two different types of repetition codes.
Let $u\in R$ be a unit and let $z\in R$ be a zero-divisor, then the code $C_u$ generated by $G_{u}=\bigg[\overbrace{uuu\cdots u}^n\bigg]$
is called the {\it unit repetition code}  and the code $C_z$ generated by 
 $G_{z}=\bigg[\overbrace{zz\cdots z}^n\bigg]$ is called the {\it zero-divisor repetition code}. Clearly $C_u$ is an 
 $(n, M, n, n)$ linear code where $M$ is the cardinality of $R$ and $C_z$ is an
 $(n, M_z, n, d_L)$ linear code where $M_z$ is the order of $z.$ 
 In particular, if we take $R = \mathbb{Z}_{p^2},$ then the unit repetition code $C_u$ 
is an $(n, p^2, n, n)$ linear code  and the zero-divisor repetition code $C_z$ is an  $(n, p, n, pn)$ linear code because of the 
order of any zero-divisor in $\mathbb{Z}_{p^2}$ is $p.$  Since $d_H>\left\lceil\frac{d_L}{p}\right\rceil,$
the code $C_u$ is type $\beta$ code. The code $C_z$ is type $\alpha$ code because  $d_H=\left\lceil\frac{d_L}{p}\right\rceil.$
 
 \begin{theorem}\label{i}
  $r_{L}(C_{z})=(p-1)n.$ 
 \end{theorem}
\begin{proof}
Let $x\in \mathbb{Z}_{p^2}^n.$ Let $w_i$ be the number of $i$ coordinates  in $x$ for $0\leq i\leq p^2-1.$
Let $\textbf{a}=\underbrace{aa\cdots a}_{n\text{ times }}.$
Consider
$d_L(x,\textbf{0})=(w_1+w_{p^2-1})+2(w_2+w_{p^2-2})+\cdots+(p-1)(w_{p-1}+w_{p^2-(p-1)})+p(w_p+w_{p+1}+\cdots+w_{p^2-p}).$
For $k=1,2,\cdots,(p-1),$ we have
$d_L(x,\textbf{kp})=(w_{kp-1}+w_{kp+1})+2(w_{kp-2}+w_{kp+2})+\cdots+p(w_0+w_1+\cdots+w_{(k-1)p}+w_{(k+1)p}+\cdots+w_{p^2-1}).$
Thus
\begin{equation*}
\begin{split}
 r_L(C_{z})&\leq \sum\limits_{k=0}^{p-1}\dfrac{d_L(x,\textbf{kp})}{p}\\
 &=\dfrac{1}{p}p(p-1)\sum\limits_{i=0}^{p^2-1}w_i\\
 &=(p-1)n.
 \end{split}
\end{equation*}
Therefore, $r_L(C_{z})\leq (p-1)n.$

On the other hand, choose $y\in \mathbb{Z}_{p^2}^n$  such that

\begin{equation*}
\sum\limits_{i=1}^{p-2}\big(p-(i+1)\big)\big(w_{i}+w_{p^2-i}\big)=\sum\limits_{i=0}^{p^2-2p}w_{p+i}
\end{equation*}
or
\begin{equation*}
\sum\limits_{i=1}^{p-2} \big(p-(i+1)\big) \big(w_{kp+i}+w_{kp-i}\big)=\sum\limits_{\substack{i=0\\i\neq kp}}^{p^2-1}w_{i}\,\, \text{for}\,
k=1,2,\cdots, (p-1)
\end{equation*}
where $w_i$ is the number of $i$ coordinates  in $y$ for $0\leq i \leq p^2-1$ and $\sum\limits_{i=0}^{p^2-1}w_i=n.$
Then 
\begin{equation*}
 \begin{split}
  d_L(y,\textbf{0})&=(w_1+w_{p^2-1})+2(w_2+w_{p^2-2})+\cdots+(p-1)(w_{p-1}+w_{p^2-(p-1)})\\
  &\hspace{2cm}+p(w_p+w_{p+1}+\cdots+w_{p^2-p})\\
&=\sum\limits_{i=1}^{p-1}i\big(w_{p+i}+w_{p-i}\big)+p\bigg(\sum\limits_{i=0}^{p^2-2p}w_{p+i}\bigg).
 \end{split}
\end{equation*}

For $k=1,2,\cdots,(p-1),$ we have
\begin{equation*}
 \begin{split}
d_L(y,\textbf{kp})&=(w_{kp-1}+w_{kp+1})+2(w_{kp-2}+w_{kp+2})+\cdots\\
&\hspace{2cm}+p(w_0+w_1+\cdots+w_{(k-1)p}+w_{(k+1)p}+\cdots+w_{p^2-1})\\
&=\sum\limits_{i=1}^{p-1} i\big(w_{kp+i}+w_{kp-i}\big)+p\bigg(\sum\limits_{\substack{i=0\\i\neq kp}}^{p^2-1}w_{i}\bigg).
\end{split}
\end{equation*}

Therefore
\begin{equation*}
 \begin{split}
  d_L(y,C_{z})&=\min \big\{d_L(y,\textbf{kp}): 0\leq k\leq (p-1)\big\}\\
  &=(p-1)n+\min \Bigg\{\sum\limits_{i=0}^{p^2-2p}w_{p+i}-\sum\limits_{i=1}^{p-2}\big(p-(i+1)\big)\big(w_{i}+w_{p^2-i}\big),\\
	&\hspace{3cm}\sum\limits_{\substack{i=0\\i\neq kp}}^{p^2-1}w_{i}-\sum\limits_{i=1}^{p-2} \big(p-(i+1)\big) \big(w_{kp+i}+w_{kp-i}\big)\Bigg\}\\
	&=(p-1)n.
 \end{split}
\end{equation*}
Therefore, $r_L(C_{z})\geq (p-1)n$ and  hence $r_L(C_{z})=(p-1)n.$
\end{proof}

\begin{theorem}\label{j}
  $r_{L}(C_{u})=(p-1)n.$
\end{theorem}
\begin{proof}
  Let $x\in \mathbb{Z}_{p^2}^n.$ Let $w_i$ be the number of $i$ coordinates in $x$ for $0\leq i\leq p^2 - 1.$ Then 
  $\sum\limits_{i=0}^{p^2-1}w_i=n.$  Consider

$d_L(x,\textbf{t})=(w_{t+1}+w_{t-1})+2(w_{t+2}+w_{t-2})+\cdots+(p-1)(w_{t-p+1}+w_{t+p-1})+p(w_{p+t}+\cdots+w_{p^2-p+t}).$

\begin{equation*}
\begin{split}
 r_L(C_{u})&\leq \frac{\sum\limits_{t=0}^{p^2-1}d_L(x,\textbf{t})}{p^2}\\
 &=(p-1)n.
 \end{split}
\end{equation*}
Therefore, $r_L(C_{u})\leq (p-1)n.$

Let $x=\overbrace{00...0}^{l}\overbrace{11\cdots1}^{l}\overbrace{22\cdots2}^{l}\cdots\overbrace{(p^2-1)(p^2-1)\cdots(p^2-1)}^{n-(p^2-1)l}\in\mathbb{Z}_{p^2}^{n}$
where $l=\bigg\lceil\dfrac{n}{p^2}\bigg\rceil,$ then $d_L(x,{\bf 0})=n+p^2(p-2)l, d_L(x,{\bf 1})=2n+p^2(p-3)l, d_L(x,{\bf2})=3n+p^2(p-4)l.$
In general, $d_L(x,{\bf i})=(i+1)n+p^2(p-(i+2))l$ for $0\leq i\leq p^{2}-2$ and $d_L(x,{\bf p^{2}-1})=p^{2}(p-1)l.$

Thus
\begin{equation*}
\begin{split}
 r_{L}(C_{u})&\geq d_L(x,C_{u})\\
&=\min\limits_{0\leq i\leq p^2-1}\bigg\{d_{L}(x,{\bf i})\bigg\}\\
&=\min\limits_{0\leq i \leq p^2-2}\bigg\{d_L(x,{\bf i}), d_L(x,{\bf p^{2}-1})\bigg\}\\
&=\min\limits_{0\leq i \leq p^2-2} \bigg\{i+1)n+p^2(p-(i+2))l, p^2(p-1)l\bigg\}.\\
\end{split}
\end{equation*}
Since $l\geq \dfrac{n}{p^2}, \,\, r_L(C_u)\geq (p-1)n.$ Therefore, $r_L(C_{u})=(p-1)n.$
\end{proof}

We define few block repetition codes over $\mathbb{Z}_{p^2}$ and find their covering radius. First, let us consider a block repetition code
over $\mathbb{Z}_{p^2}$ with generator matrix 
$$G=\bigg[\overbrace{11\cdots1}^n\overbrace{22\cdots2}^n\cdots \overbrace{(p^2-1)(p^2-1)\cdots (p^2-1)}^n\bigg].$$
Let $1\leq i< p.$  Since the subgroup $\langle ip \rangle$ generated by $ip$ is of order $p,$ each element of $\langle ip \rangle$ appears $p$ times
in $ip\{0, 1, 2, \cdots, (p^2-1)\}.$ 
 
 Therefore, $ip(12\cdots (p^2-1))$ contains $p-1$ zeros and remaining $p^2-1-(p-1)$ are zero-divisors. Since Lee weight of zero-divisor is $p,$ 
  $wt_L(ip(1 2 \cdots (p^2-1)))=(p^2-p)p$ and hence $wt_L(ip({\bf 1} {\bf 2} \cdots {\bf p^2-1}))=n(p^2-p)p=np^2(p-1)$ where ${\bf i}\in\mathbb{Z}_{p^2}^n.$

If $x\neq ip$ is not a zero-divisor, then $x\{1, 2, \cdots, p^2-1\}=\{1, 2, \cdots, p^2-1\},$

$wt_L(x(1 2 \cdots (p^2-1)))=wt_L(1 2 \cdots (p^2-1))=p^2(p-1)$ and hence

$wt_L(x({\bf 1} {\bf 2} \cdots {\bf (p^2-1)}))=$ $np^2(p-1).$
Therefore, the code is $(n(p^2-1), p^2, np^2(p-1))$ a constant weight code.

If we consider the Hamming weight of the code, then the code is two weight code of weights $(p^2-p)n$ and $(p^2-1)n.$
Thus we have
\begin{theorem}

 The code  generated by $G=\bigg[\overbrace{11\cdots1}^n\overbrace{22\cdots2}^n\cdots \overbrace{(p^2-1)\cdots (p^2-1)}^n\bigg]$ is 
 a $\big((p^2-1)n, p^2, p(p-1)n, p^2(p-1)n\big)$ $\mathbb{Z}_{p^2}$-linear code with weight distributions with respect to the Hamming and with respect
 to the Lee is $A_0=1, A_{(p^2-p)n}=p-1$ and $A_{(p^2-1)n}=p^2-p.$.
\end{theorem}

The above code is called a {\it block repetition code} and  is denoted by $BR_{p^2}^{(p^2 - 1)n}.$ 
Since $d_H= \bigg\lceil\dfrac{d_L}{p}\bigg\rceil,$ this block repetition code is type $\alpha$ code.

 \begin{theorem} \label{k}
 $(p^3-p^2-1)n\leq r_{L}(BR_{p^2}^{(p^2-1)n})\leq (p^3-p^2)n.$
\end{theorem}
\begin{proof}
Let $x=11\cdots1\in \mathbb{Z}_{p^2}^{(p^2-1)n},$ then we have 
$d_{L}(x,BR_{p^2}^{(p^2-1)n})=(p^3-p^2-1)n.$ Hence by definition of covering radius, $r_{L}(BR_{p^2}^{(p^2-1)n})\geq (p^3-p^2-1)n.$
On the other hand, its Gray image $\phi(BR_{p^2}^{(p^2-1)n})$ contains a codeword 

$$y=\overbrace{11\cdots1}^{p^2n}\overbrace{22\cdots 2}^{p^2n}\cdots\overbrace{(p-1)(p-1)\cdots (p-1)}^{p^2n}\overbrace{00\cdots0}^{p(p-1)n}.$$

Let $C_1$ be the code generated by $y.$ Let $C_2$ be the code with generator matrix 
$$G_2=\bigg[\overbrace{11\cdots1}^{p^2n}\overbrace{22\cdots 2}^{p^2n}\cdots\overbrace{(p-1)(p-1)\cdots (p-1)}^{p^2n}\bigg]$$
and let $C_3=\{\overbrace{00\cdots0}^{p(p-1)n}\}.$ Then $C_1$ is the cartesian product $C_2\times C_3$ of $C_2$ and $C_3.$  Since
the code $C_2$ is equivalent to the repetition code $[(p-1)p^2n, 1, (p-1)p^2n]$ over $\mathbb{Z}_p,$ the covering radius of $C_2$ is
$\bigg\lceil\dfrac{(p-1)p^2n(p-1)}{p} \bigg\rceil=p(p-1)^2n.$ By Theorem \ref{cb}, 
the covering radius of $C_1$ is $p(p-1)n+p(p-1)^2 n=(p^3-p^2)n.$
\begin{align*}
 r_H(C_2\times C_3)&\leq\bigg\lceil\dfrac{(p-1)p^2n(p-1)}{p}\bigg\rceil+p(p-1)n\\
 &=p(p-1)^2n+p(p-1)n\\
 &=p(p-1)((p-1)+1)n\\
 &=(p^3-p^2)n.
\end{align*}
Since $C_1\subset \phi(BR_{p^2}^{(p^2-1)n}).$ This implies that, 
\begin{align*}
r_H(BR_{p^2}^{(p^2-1)n})&\leq r_H(C_1)\\
&=r_H(C_2\times C_3)\\
&=(p^3-p^2)n.
\end{align*}

Therefore $r_L(BR_{p^2}^{(p^2-1)n})\leq (p^3-p^2)n.$
\end{proof}

Now, we define another block repetition code over $\mathbb{Z}_{p^2}.$ Let 
$$G=\bigg[\overbrace{11\cdots1}^n\overbrace{22\cdots2}^n\cdots \overbrace{(p^2-2)(p^2-2)\cdots (p^2-2)}^n\bigg].$$
Then the code generated by $G$ is a $\left((p^2-2)n, p^2, pn, p^2 n\right)$ linear code over $\mathbb{Z}_{p^2}.$

The following theorem gives an upper and a lower bounds of the covering radius of this block repetition code.
 \begin{theorem}\label{l}
 $(p^3-p^2-2)n\leq r_{L}(BR_{p^2}^{(p^2-2)n})\leq (p^3-p^2-1)n.$ 
\end{theorem}
A proof of this theorem is similar to that of Theorem \ref{k} and is omitted.

All block repetition codes discussed above are codes with the same length of the block. Now we are going to discuss a block repetition code with different 
block length. 

Let
$$G=\bigg[\overbrace{11 \cdots 1}^{m}\overbrace{pp \cdots p}^{n}\overbrace{2p2p \cdots 2p}^{n} \cdots \overbrace{(p-1)p \cdots (p-1)p}^{n}\bigg].$$
Then the $wt(p({\bf 1}{\bf p}{\bf 2p}\cdots {\bf (p-1)p}))=wt({\bf p}{\bf 0}\cdots{\bf 0})$ and hence $wt_H(C)=m$ and $wt_L(C)=mp.$
Therefore the code generated by $G$ is a $\big(m+(p-1)n, p^2, m, pm\big)$  linear code over $\mathbb{Z}_{p^2}.$

\begin{theorem}\label{m}
 $m+(p^2-p-1)n\leq r_{L}(BR_{p^2}^{m+(p-1)n})\leq m+(p^2-p)n.$ 
\end{theorem}
\begin{proof}
Let $x=\overbrace{00...0}^{m}\overbrace{11\cdots1}^{(p-1)n}\in \mathbb{Z}_{p^2}^{m+(p-1)n},$ then we have 
$d_{L}(x,BR_{p^2}^{m+(p-1)n})=m+(p^2-p-1)n.$ Hence by definition, $r_{L}(BR_{p^2}^{m+(p-1)n})\geq m+(p^2-p-1)n.$ On the other hand,
its Gray image $\phi(BR_{p^2}^{m+(p-1)n})$ is equivalent to $p$-ary linear code  with generator matrix
\begin{center}
$\begin{pmatrix}
 A&\vline&0\\
 \hline
 B&\vline&C
\end{pmatrix}$
\end{center}

where
\begin{align*}
 A&=\begin{bmatrix}
  \overbrace{11\cdots 1}^{m+np}|\overbrace{22\cdots2}^{np}|\cdots|\overbrace{(p-1)(p-1)\cdots (p-1)}^{np}
 \end{bmatrix},\\
  B&=\begin{bmatrix}
  \overbrace{11\cdots1000}^{m+np}|\overbrace{00\cdots0}^{np}|\cdots|\overbrace{00\cdots0}^{np}
 \end{bmatrix}\\
 \text{ and }  C&=\begin{bmatrix}
  \overbrace{111\cdots11}^{(p-1)m}
 \end{bmatrix}.
\end{align*}
\end{proof}

\section{Conclusion}
We have computed lower and upper bounds on the covering radius of codes over
$\mathbb{Z}_{p^2}$ with respect to Lee distance. We also computed the covering radius of various repetition codes over $\mathbb{Z}_{p^2}.$
It would be an interesting future task to find out the exact covering radii of many of these codes and generalize the results for codes over
$\mathbb{Z}_{p^s}.$

\end{document}